\documentclass[11pt]{article}
\usepackage{mathrsfs,amssymb,amsmath,amsthm,amsfonts}
\usepackage{graphicx,subfig,float}

\usepackage{color}
\usepackage[colorlinks, linkcolor=red, anchorcolor=green, citecolor=blue]{hyperref}
\allowdisplaybreaks

\newtheorem{theorem}{Theorem}[section]
\newtheorem{lemma}[theorem]{Lemma}

\newtheorem{remark}[theorem]{Remark}
\newtheorem{definition}[theorem]{Definition}
\newtheorem{example}[theorem]{Example}

\numberwithin{equation}{section}
\begin{document}

\title{\Large\bf Signal Recovery under Cumulative Coherence}

\author{Peng Li\,$^\dag$ and Wengu Chen
\thanks{Corresponding author}
\thanks{P. Li is with  Graduate School, China Academy of Engineering Physics, Beijing 100088, China (E-mail: lipengmath@126.com)}
\thanks{W. Chen is with Institute of Applied Physics and Computational Mathematics, Beijing 100088, China (E-mail: chenwg@iapcm.ac.cn)}
}
\date{ }

\maketitle
\textbf{Abstract}~~This paper considers signal recovery in the framework of cumulative coherence. First, we show that the Lasso estimator and the Dantzig selector exhibit similar behavior under the cumulative coherence. Then we estimate the approximation equivalence between the Lasso and the Dantzig selector by calculating prediction loss difference under the condition of cumulative coherence.  And we also prove that the cumulative coherence implies the restricted eigenvalue condition. Last,
we illustrate the advantages of cumulative coherence condition for three class matrices, in terms of the recovery performance of sparse signals via extensive numerical experiments.



%%%%%%%%%%%%%%%%%%%%%%%%%%%%%%%%%%%%%%%%%%%%%%%%%%%%%%%%%%%%%%%%%%%%%%%%%%%%%%%%%%%%%%%%%%%%%%%%

%%%%%%%%%%%%%%%%%%%%%%%%%%%%%%%%%%%%%% Section 1 %%%%%%%%%%%%%%%%%%%%%%%%%%%%%%%%%%%%%%%%%%%%%%%

%%%%%%%%%%%%%%%%%%%%%%%%%%%%%%%%%%%%%%%%%%%%%%%%%%%%%%%%%%%%%%%%%%%%%%%%%%%%%%%%%%%%%%%%%%%%%%%%
\section{Introduction \label{s1}}
\hskip\parindent

Compressed sensing predicts that sparse signals can be reconstructed from what was previously believed to be incomplete information.  Since Cand\`{e}s, Romberg and Tao's seminal works \cite{CRT2006,CRT2006-1} and Donoho's ground-breaking work \cite{D2006},  this new field has triggered a large research in mathematics, engineering and medical image. In such contexts, we often require to recover an unknown signal $x\in\mathbb{R}^n$ from an underdetermined system of linear equations
\begin{align}\label{systemequationsnoise}
b=Ax+z,
\end{align}
where $b\in\mathbb{R}^m$ are available measurements, the matrix $A\in\mathbb{R}^{m\times n}~(m<n)$ models the linear measurement process and
$z\in\mathbb{R}^m$ is a vector of measurement errors.

For the reconstruction of $x$, the most intuitive approach is to find the sparsest signal in the feasible set of possible solutions, which leads to an $\ell_0$-minimization problem as follows
$$
\min_{x\in\mathbb{R}^n}\|x\|_0~~\text{subject~ to}~~b-Ax\in\mathcal{B},
$$
where $\|x\|_0$ denotes the $\ell_0$ norm of $x$, i.e., the number of nonzero coordinates, and $\mathcal{B}$ is a bounded set determined by the error structure. However, such method is NP-hard and thus computationally infeasible in high dimensional sets.  Cand\`{e}s and Tao \cite{CT2005} proposed a convex relaxation of this method-the constrained $\ell_1$ minimization method. It estimates the signal $x$ by
\begin{align}\label{BPmodel}
\min_{x\in\mathbb{R}^n}~\|x\|_1~~\text{subject~ to}~~Ax=b,
\end{align}
which is also called basis pursuit (BP) \cite{CDS1998}.

When $z\neq 0$, i.e., there exist noises, we often consider two cases. One is $l_2$ bounded noises \cite{DET2006}, i.e.,
\begin{align}\label{QCBPmodel}
\min_{x\in\mathbb{R}^n}~\|x\|_1 ~~\text{subject~ to}~~\|b-Ax\|_{2}\leq\eta
\end{align}
for some constant $\eta>0$, which is called quadratically constrained basis pursuit (QCBP). And the other is motivated by \textit{Dantzig~selector} procedure \cite{CT2007}, i.e.,
\begin{align}\label{DantzigselectorMedel}
\min_{x\in\mathbb{R}^n}~\|x\|_1 ~~\text{subject~ to}~~\|A^*(b-Ax)\|_{\infty}\leq\eta.
\end{align}

For the $\ell_1$-minimization problem (\ref{BPmodel}) or (\ref{QCBPmodel})-(\ref{DantzigselectorMedel}), there are many works under null space property (NSP) introduced by Donoho and Elad \cite{DE2003}. We say that the measurement matrix $A$  satisfies the NSP, if there exists a constant $0<\tau<1$ such that
$$
\|x_{\max(s)}\|_1\leq\tau\|x_{-\max(s)}\|_1,
$$
where and in what follows, $x_{\max(s)}$ is the vector $x$ with all but the largest $s$ entries in absolute value set to zero, and $x_{-\max(s)}=x-x_{\max(s)}$.  There are many works under NSP, readers can refer to \cite{DE2003, CDD2009, S2011, FR2013, F2014}. For the $\ell_1$-minimization problem (\ref{BPmodel}) or (\ref{QCBPmodel})-(\ref{DantzigselectorMedel}), there are also many works under the restricted isometry property (RIP) introduced in \cite{CT2005}. For a matrix $A\in\mathbb{R}^{m\times n}$, $s\in[[1,n]]$, the $s$-th restricted isometry constant $\delta_s=\delta_s(A)$  is the smallest number such that
$$
(1-\delta_s)\|x\|_2\leq\|Ax\|_2^2\leq(1+\delta_s)\|x\|_2^2
$$
for all $x\in\mathbb{R}^n$ with $\|x\|_0\leq s$ . We say that the matrix $A$ satisfies the restricted isometry property if $\delta_s$ is small for reasonably large $s$.  There are many works under this condition, readers can refer to
\cite{CT2005,D2006,CT2006,CT2007,CDD2009, CZ2013-1,CZ2014,ZL2017}. What is worth mentioning is that Cai and Zhang \cite{CZ2013} established a sharp condition $\delta_s+\theta_{s,s}<1$ about restricted isometry constant $\delta_s$ and restricted orthogonality constant $\theta_{s,s}$ for $s$-sparse signal's exact recovery. And they also showed that the condition is sufficient to guarantee the stable recovery for the noisy case.

Here we consider recovering a signal under the framework of cumulative coherence,  a regularity and widely used condition.  Since the cumulative coherence is a generalization of the coherence, we first recall the coherence, which was introduced by Donoho and Huo in \cite{DH2001}.

\begin{definition}\label{MIPDefinition}
Let $A\in\mathbb{R}^{m\times n}$ be a matrix with $\ell_2$-normalized columns $A_1,\ldots,A_n$, i.e., $\|A_i\|_2=1$ for all $i=1,\ldots,n$. The coherence $\mu=\mu(A)$ of matrix $A$ is defined as
\begin{align*}
\mu=\max_{1\leq i\neq j\leq n}|\langle A_i, A_j \rangle|.
\end{align*}
When coherence $\mu$ is small, we say that $A$ satisfies mutual incoherence property (MIP).
\end{definition}

It was first shown by Donoho and Huo \cite{DH2001}, in the noiseless case for the setting where $A$ is a concatenation of two square orthogonal matrices, that $\mu<1/(2s-1)$ ensures the exact recovery of $x$ when $x$ is $s$-sparse. And in the noisy case, Donoho, Elad and Temlyakov \cite{DET2006} showed that sparse signals can be recovered approximately via (\ref{QCBPmodel}) with the error at worst proportional to
the input noise level, under the condition $\mu<1/(4s-1)$. Cai, Wang and Xu \cite{CWX2010} showed that the MIP condition $\mu<1/(2s-1)$ is sharp
for exact recovery of $s$-sparse signals and also got the stable recovery via QCBP model (\ref{QCBPmodel}) and Dantzig selector (\ref{DantzigselectorMedel}) under this condition. More results under the framework of MIP, readers can refer to \cite{T2004-1,T2004-2,T2006,CXZ2009,T2009, X2011, XX2015}.

The coherence parameter does not characterize a measurement matrix very well since it only reflects
the most extreme correlations between columns. When most of the inner products are tiny, the coherence can be downright
misleading. A wavelet packet dictionary exhibits this type of behavior. To remedy this shortcoming, Tropp \cite{T2004-1,T2004-2} introduced the cumulative coherence function, which measures the maximum total coherence between a fixed column and a collection of other columns. This cumulative coherence is a generalization of cumulative coherence, which incorporates the usual coherence as the particular value $s=1$ of its argument.

\begin{definition}\label{CCDefinition}
Let $A\in\mathbb{R}^{m\times n}$ be a matrix with $\ell_2$-normalized columns $A_1,\ldots,A_n$, i.e., $\|A_i\|_2=1$ for all $i=1,\ldots,n$. The cumulative coherence function $\mu_1(s)=\mu_1(A,s)$ of matrix $A$ is defined for $s\in[n-1]$ by
\begin{align*}
\mu_1(s)=\max_{\substack{S\subset\{1,\ldots,n\}\\|S|\leq n}}\max_{i\in S^c}\sum_{j\in S}|\langle A_i, A_j \rangle|.
\end{align*}
When the cumulative coherence of a matrix grows slowly, we say informally that the dictionary is quasi-incoherent.
\end{definition}

In \cite{T2004-2}, Tropp showed that
\begin{align}
\mu_1(s-1)+\mu_1(s)<1
\end{align}
can guarantee that all $s$-sparse signals can be recovered exactly through basis pursuit model (\ref{BPmodel}) or orthogonal matching pursuit. Tropp also gave an example to demonstrate how much the cumulative coherence
function improves on the coherence parameter.
Since then, cumulative coherence was studied by many scholars. For example, Schnass and Vandergheynst \cite{SV2008} gave out the Welch-type bound for the cumulative coherence function.
Herrity, Gilbert and Tropp \cite{HGT2006}, Foucart and Rauhut \cite[Section 5.5]{FR2013} analysed the thresholding algorithms under the condition cumulative coherence.
We also notice that there exists close relationship between coherence, RIP and cumulative coherence. As showed in \cite{T2004-2}, $mu\leq\mu_1(s)\leq s\mu$, and in \cite{FR2013}, Foucart and Rauhut also showed that $\mu_1(s-1)\geq \delta_s$.
More works about cumulative coherence, readers can see \cite{SV2007,L2010,DGN2012}. Because cumulative coherence has a property which is similar to the definition of restricted orthogonality constant (see Lemma \ref{OrthogonalCC}), and motivated by Cai and Zhang's work \cite{CZ2013}, we consider the cumulative coherence analysis of stable recovery via QCBP and Dantzig selector.

Instead of solving (\ref{QCBPmodel}) directly, many authors also studied the following unconstrained Lasso problem
\begin{align}\label{LassoModel}
\min_{x\in\mathbb{R}^n}~\lambda\|x\|_1+\frac{1}{2}\|Ax-b\|_2^2,
\end{align}
where $\lambda\geq 0$. We point out this problem was first introduced in \cite{T1996}.
There are many works about this model. For example, Lin and Li \cite{LL2014}, Shen, Han and Braverman \cite{SHB2015} showed that $x$ can be stably recovered via analysis based approaches under the restricted isometry property.
And Xia and Li \cite{XL2016} got the error bounds in the analysis Lasso by restricted eigenvalue condition under a sparsity scenario and by the $\ell_2$ robust null space property under a non-sparsity scenario. Readers can refer to \cite{DET2006,EMR2007,BRT2009,C2013,TEBN2014,ZYY2016} to see more works about Lasso model. But as far as we know, there lacks MIP or cumulative coherence based theoretical study about Lasso.

In this paper, we purse the cumulative coherence analysis of  QCBP model (\ref{QCBPmodel}), Dantzig selector model (\ref{DantzigselectorMedel}) and Lasso model (\ref{LassoModel}).
Our contributions of this paper can be stated as follows.

First, we show that cumulative coherence function has a property which is similar to the definition of restricted
orthogonality constant (Lemma \ref{OrthogonalCC}). And using the sparse representation of a polytope in \cite[Lemma 1.1]{CZ2014}, we also estimate a key technical tool, which provides a estimate the inner product $\langle Ax, Ay\rangle$ by cumulative coherence function when only one component is sparse (Lemma \ref{NonsparseCC}). Then by this useful tool, we show that the condition $\mu_1(s-1)+\mu_1(2s-1)<1$ is sufficient to guarantee the stable recovery via the constrained $\ell_1$ minimization (\ref{QCBPmodel}) or (\ref{DantzigselectorMedel}) (Theorem \ref{QCBP-DS}). We also show that the condition $\mu_1(s-1)+\mu_1(4s-1)<1/\sqrt{3}$ is sufficient to guarantee the stable recovery via the unconstrained minimization model (\ref{LassoModel}) (Theorem \ref{Lasso}). What should be pointed out is that it is the first time to give the cumulative coherence analysis for the noisy case for $\ell_1$ minimization, especially for Lasso model (\ref{LassoModel}). And our result improves the condition $\mu<1/(4s-1)$ for QCBP model in Donoho, Elad, and Temlyakov \cite{DET2006}, the condition $\mu_1(2s)<1/2$ for iterative hard thresholding and the condition $\mu_1(s-1)+2\mu_1(s)<1$ for hard thresholding pursuit in \cite[Chapter 5]{FR2013}.

Then in Section \ref{s3}, we estimate the closeness of this prediction loss $\|A\hat{x}^{DS}-Ax\|_2^2$ and $\|A\hat{x}^{L}-Ax\|_2^2$ in the framework of cumulative coherence, where $\hat{x}^{DS}$ and $\hat{x}^{L}$ are the minimizers of Dantzig selector model (\ref{DantzigselectorMedel}) and Lasso model (\ref{LassoModel}), respectively.
We get an oracle inequality for sparse signal via Dantzig selector with Gaussian noise under the cumulative coherence in Section \ref{s4}. And in Section \ref{s5}, we investigate relationship between cumulative coherence and restricted eigenvalue condition, and find that the restricted eigenvalue condition can be deduced from the cumulative coherence.

Last, in Section \ref{s6}, we illustrate the advantages of cumulative coherence condition in terms of the recovery performance of sparse signals via extensive numerical experiments. We compute the unconstraint problem Lasso (\ref{LassoModel}) through the IRucLq-v method in \cite{LXY2013},  for three different measurement matrices-decaying matrix, Dirac-Hadamard matrix and Dirac-Fourier matrix. These matrices satisfy the cumulative coherence condition proposed in this paper.

Throughout the article, we use the following basic notation. Let $x_S$ be the vector equal to $x$ on $S$ and to zero on $S^c$. For any positive integer $n$, let $[[1,n]]$ denote the set $\{1,\ldots,n\}$. We also let a vetor $u\in\mathbb{R}^n$  denote an ``indicator vector", i.e., it has only one non-zero entry and the value of this entry is either 1 or -1.

%%%%%%%%%%%%%%%%%%%%%%%%%%%%%%%%%%%%%%%%%%%%%%%%%%%%%%%%%%%%%%%%%%%%
%%%%%%%%%%%%%%%%%%%%%%   section 2 %%%%%%%%%%%%%%%%%%%%%%%%%%%%%%%%%
%%%%%%%%%%%%%%%%%%%%%%%%%%%%%%%%%%%%%%%%%%%%%%%%%%%%%%%%%%%%%%%%%%%
\section{Stable Recovery \label{s2}}
\hskip\parindent

Now, we consider the stable recovery of signals $x\in\mathbb{R}^{n}$ through QCBP model (\ref{QCBPmodel}), Dantzig selector model (\ref{DantzigselectorMedel}) and Lasso model (\ref{LassoModel}). First, we establish two useful properties of cumulative coherence in Subsection \ref{s2.1}. And then we show that the condition $\mu_1(s-1)+\mu_1(2s-1)<1$ is sufficient for the stably recovery through QCBP model (\ref{QCBPmodel}) and Dantzig selector model (\ref{DantzigselectorMedel}) in Subsection \ref{s2.2}.
And in Subsection \ref{s2.3}, we give out a sufficient condition $\mu_1(s-1)+\mu_1(4s-1)<1/\sqrt{3}$, which can guarantee the stable recovery via Lasso model (\ref{LassoModel}).

\subsection{Properties of Cumulative Coherence \label{s2.1}}
\quad

In this subsection, we will give several useful properties of cumulative coherence.
The first one, which is similar to the definition of the $s$-th restricted isometry constant \cite{CT2005}, comes from \cite[Theorem 5.3]{FR2013}.
\begin{lemma}\label{CCLemma}
Let $A\in\mathbb{R}^{m\times n}$ be a matrix with $\ell_2$-normalized columns and $s\in\{1,\ldots,n\}$. For all $s$-sparse vectors $x\in\mathbb{R}^n$,
\begin{align*}
\big(1-\mu_1(s-1)\big)\|x\|_2^2\leq\|Ax\|_2^2\leq\big(1+\mu_1(s-1)\big)\|x\|_2^2.
\end{align*}
\end{lemma}

And the second one provides a way to estimate the inner product $\langle Ax, Ay\rangle$ by cumulative coherence function $\mu_1$ when both component are  sparse. This property is similar to the definition of the $(s,t)$-restricted orthogonality constant $\theta_{s,t}$ \cite{CT2005}.
\begin{lemma}\label{OrthogonalCC}
Suppose that $x$ is $s$-sparse and $y$ is $t$-sparse, then
$$
\big|\langle Ax, Ay \rangle-\langle x,y\rangle\big|\leq\mu_1(s+t-1)\|x\|_2\|y\|_2.
$$
And moreover, if $\text{supp}(x)\cap\text{supp}(y)=\emptyset$, then
$$
\big|\langle Ax, Ay \rangle\big|\leq\mu_1(s+t-1)\|x\|_2\|y\|_2.
$$
\end{lemma}
\begin{proof}
Our proof follows the idea of \cite[Lemma II.2]{LLS2013}.

Suppose that $\text{supp}(x)\subset S$ with $|S|=s$ and $\text{supp}(y)\subset T$ with $|T|=t$, and assume that $\|x\|_2=\|y\|_2=1$. We consider the following identity
\begin{align}\label{e2.1}
\langle Ax, Ay\rangle=\frac{1}{4}\big(\|A(x+y)\|_2^2-\|A(x-y)\|_2^2\big).
\end{align}
Note that $\|x+y\|_0=\|x-y\|_0\leq s+t$. According to Lemma \ref{CCLemma}, we have
\begin{align*}
\langle Ax, Ay\rangle&\geq \frac{1}{4}\Big(\big(1-\mu_1(s+t-1)\big)\|x+y\|_2^2-\big(1+\mu_1(s+t-1)\big)\|x-y\|_2^2\Big)\\
&=\langle x,y\rangle-\frac{1}{2}\mu_1(s+t-1)\big(\|x\|_2^2+\|y\|_2^2\big)\\
&=\langle x,y \rangle-\mu_1(s+t-1),
\end{align*}
i.e.,
\begin{align}\label{e2.2}
\langle Ax, Ay\rangle-\langle x,y \rangle\geq-\mu_1(s+t-1).
\end{align}

On the other hand,
\begin{align*}
\langle Ax, Ay\rangle&\leq \frac{1}{4}\Big(\big(1+\mu_1(s+t-1)\big)\|x+y\|_2^2-\big(1-\mu_1(s+t-1)\big)\|x-y\|_2^2\Big)\\
&=\langle x,y\rangle+\frac{1}{2}\mu_1(s+t-1)\big(\|x\|_2^2+\|y\|_2^2\big)\\
&=\langle x,y \rangle+\mu_1(s+t-1),
\end{align*}
i.e.,
\begin{align}\label{e2.3}
\langle Ax, Ay\rangle-\langle x,y \rangle\leq\mu_1(s+t-1).
\end{align}
Combining the inequality (\ref{e2.2}) and the inequality (\ref{e2.3}), we have
\begin{align}\label{e2.4}
|\langle Ax, Ay\rangle-\langle x,y \rangle|\leq\mu_1(s+t-1).
\end{align}

For general $x$ and $y$, we consider $x'=x/\|x\|_2$ and $y'=y/\|y\|_2$. Then by (\ref{e2.4}),  we have
\begin{align*}
\bigg|\langle A\frac{x}{\|x\|_2}, A\frac{y}{\|y\|_2}\rangle-\langle \frac{x}{\|x\|_2},\frac{y}{\|y\|_2} \rangle\bigg|
=|\langle Ax', Ay'\rangle-\langle x',y' \rangle|\leq\mu_1(s+t-1),
\end{align*}
which implies that
\begin{align*}
|\langle Ax, Ay\rangle-\langle x,y \rangle|\leq\mu_1(s+t-1)\|x\|_2\|y\|_2.
\end{align*}
\end{proof}

Now, we can give out the key technical tool used in the main results. It provides a way to estimate the inner product $\langle Ax, Ay\rangle$ by cumulative coherence function $\mu_1$ when only one component is sparse. Our idea is inspired by \cite[Lemma 5.1]{CZ2013}.

\begin{lemma}\label{NonsparseCC}
Let $s_1, s_2\leq n$ and $\alpha\geq 0$. Suppose $x,y\in\mathbb{R}^n$ satisfies $\text{supp}(x)\cap \text{supp}(y)=\emptyset$ and $x$ is $s_1$ sparse. If $\|y\|_1\leq\alpha s_2$ and $\|y\|_\infty\leq\alpha$, then
\begin{align}\label{e2.5}
|\langle Ax, Ay \rangle|\leq\alpha \sqrt{s_2}\mu_1(s_1+s_2-1)\|x\|_2.
\end{align}
\end{lemma}
\begin{proof}
Suppose $\|y\|_0=t$. We consider two cases as follows.

\textbf{Case I: $t\leq s_2$.}

By Lemma \ref{OrthogonalCC} and $\|y\|_\infty\leq\alpha$, we have
\begin{align*}
|\langle Ax, Ay \rangle|&\leq\mu_1(s_1+s_2-1)\|x\|_2\|y\|_2\leq\mu_1(s_1+s_2-1)\|x\|_2\|y\|_{\infty}\sqrt{\|y\|_0}\\
&\leq\alpha\sqrt{t}\mu_1(s_1+s_2-1)\|x\|_2\leq\alpha \sqrt{s_2}\mu_1(s_1+s_2-1)\|x\|_2.
\end{align*}

\textbf{Case II: $t> s_2$.}

We shall prove by induction. Assume that (\ref{e2.5}) holds for $t-1$.
For any $0<p\leq\infty$, let $\mathcal{B}^{\ell_p}(r)=\{u\in\mathbb{R}^n: \|u\|_p\leq r\}$. The condition $\|y\|_1\leq\alpha s_2$ and $\|y\|_{\infty}\leq \alpha$ imply that $y\in s_2\mathcal{B}^{\ell_1}(\alpha)\cap\mathcal{B}^{\ell_{\infty}}(\alpha)$.  By sparse representation of
a polytope in \cite[Lemma 1.1]{CZ2014}, $y$ can be represented as the convex hull of $s_2$-sparse vectors:
$$
y=\sum_{j=1}^N\gamma_ju^j,
$$
where $u^j$ is $s_2$-sparse for all $j\in[N]$ and
\begin{align*}
\sum_{j=1}^N\gamma_j=1,~~0<\gamma_j\leq 1, j\in[N].
\end{align*}

Since $u^j$ is $s_2$-sparse and $s_2\leq t-1$, we can use the induction assumption,
\begin{align*}
|\langle Ax, Ay\rangle|&\leq\sum_{j=1}^N\gamma_j|\langle Ax, Au^j\rangle|\leq\sum_{j=1}^N\gamma_j\bigg(\alpha\sqrt{s_2}\mu_1(s_1+s_2-1)\|x\|_2\bigg)\\
&=\alpha\sqrt{s_2}\mu_1(s_1+s_2-1)\|x\|_2,
\end{align*}
which gives (\ref{e2.5}) for $t$.
\end{proof}

\subsection{Stable Recovery for Dantzig Selector and QCBP \label{s2.2}}
\hskip\parindent

In this subsection, we consider the stable recovery of signals through Dantzig selector model (\ref{DantzigselectorMedel}) and QCBP model (\ref{QCBPmodel}).

\begin{theorem}\label{QCBP-DSgeneral}
Assume the cumulative coherence function of measurement matrix $A$ satisfies the condition
\begin{align}\label{QCBP-DScondition}
\mu_1(a-1)+C_{a,s}\mu_1(2s-1)<1,
\end{align}
where $C_{a,s}=\sqrt{(2s-a)/a}$ and $1\leq a\leq s$.
Let $\hat{x}^{DS}$ be the solution to the Dantzig selector (\ref{DantzigselectorMedel}), then
\begin{align*}
\|\hat{x}^{DS}-x\|_2&\leq\frac{2\sqrt{2}\sqrt{s}}{1-\mu_1(a-1)-C_{a,s}\mu_1(2s-1)}\eta\\
&\hspace*{12pt}+\bigg(\frac{\sqrt{2}s\mu_1(2s-1)/a}{\big(1-\mu_1(a-1)-C_{a,s}\mu_1(2s-1)\big)C_{a,s}}+1\bigg)\frac{2\|x_{-\max(s)}\|_1}{\sqrt{s}}.
\end{align*}
Let $\hat{x}^{\ell_2}$ be the solution to the QCBP (\ref{QCBPmodel}), then
\begin{align*}
\|\hat{x}^{\ell_2}-x\|_2&\leq\frac{2\sqrt{2}\sqrt{1+\mu_1(a-1)}}{1-\mu_1(a-1)-C_{a,s}\mu_1(2s-1)}\eta\\
&\hspace*{12pt}+\bigg(\frac{\sqrt{2}s\mu_1(2s-1)/a}{\big(1-\mu_1(a-1)-C_{a,s}\mu_1(2s-1)\big)C_{a,s}}+1\bigg)\frac{2\|x_{-\max(s)}\|_1}{\sqrt{s}}.
\end{align*}
\end{theorem}

To prove Theorem \ref{QCBP-DSgeneral}, we need the following two auxiliary lemmas. The first one is the cone constraint inequality, which comes from \cite[Page 1215]{CRT2006} for $\hat{x}^{\ell_2}$, and \cite[Page 2330]{CT2007} for $\hat{x}^{DS}$.

\begin{lemma}\label{ConeconstriantinequalityLemma}
The minimization solutions $\hat{x}^{DS}$ of (\ref{DantzigselectorMedel}) and $\hat{x}^{\ell_2}$ of (\ref{QCBPmodel}) satisfy
\begin{align*}
\|h_{-\max(s)}\|_1\leq \|h_{\max(s)}\|_1+2\|x_{-\max(s)}\|_1,
\end{align*}
where $h=\hat{x}^{DS}-x$ or $h=\hat{x}^{\ell_2}-x$.
\end{lemma}

The second one comes from \cite{CZ2013-1}, which will be used to estimate $\|(\hat{x}-x)_{-\max(s)}\|_2$.
\begin{lemma}\label{ApproximatingerroreEtimate}
Suppose $n\geq s$, $c_1\geq c_2\geq\cdots\geq c_n\geq 0$, $\sum_{j=1}^sc_j\geq\sum_{j=s+1}^nc_j$, then for all $\beta\geq 1$,
$$
\sum_{j=s+1}^nc_j^{\beta}\leq\sum_{j=1}^sc_j^{\beta}.
$$
Moreover generally, suppose $c_1\geq c_2\geq\cdots\geq c_n\geq 0$, $\rho\geq 0$ and $\sum_{j=1}^sc_j+\rho\geq\sum_{j=s+1}^nc_j$, then for all $\beta\geq 1$,
$$
\sum_{j=s+1}^nc_j^{\beta}\leq s\Bigg(\sqrt[\beta]{\frac{\sum_{j=1}^sc_j^{\beta}}{s}}+\frac{\rho}{s}\Bigg)^{\beta}.
$$
\end{lemma}

Now, we have made preparations for proving Theorem \ref{QCBP-DSgeneral}.
\begin{proof}[Proof of Theorem \ref{QCBP-DSgeneral}]
Let $h=\hat{x}^{DS}-x$. We have the following tube constraint inequality
\begin{align}\label{Tubeconstraintinequality}
\|A^*Ah\|_{\infty}\leq\|A^*(A\hat{x}^{DS}-b)\|_{\infty}+\|A^*(b-Ax)\|_{\infty}\leq\eta+\eta=2\eta.
\end{align}
By Lemma \ref{ConeconstriantinequalityLemma}, we have cone constraint inequality as follows
\begin{align}\label{Coneconstraintinequality}
\|h_{-\max(s)}\|_1\leq \|h_{\max(s)}\|_1+2\|x_{-\max(s)}\|_1.
\end{align}

By
\begin{align*}
\|h\|_2=\sqrt{\|h_{\max(s)}\|_2^2+\|h_{-\max(s)}\|_2^2},
\end{align*}
we need to estimate $\|h_{\max(a)}\|_2$ and $\|h_{-\max(s)}\|_2$, respectively.

By $\|h_{\max(s)}\|_2\leq\sqrt{s/a}\|h_{\max(a)}\|_2$, it suffices to estimate $\|h_{\max(a)}\|_2$.  we consider the following identity
\begin{align}\label{e2.8}
\big|\langle Ah, Ah_{\max(a)}\rangle\big|=\big|\langle Ah_{\max(a)},Ah_{\max(a)}\rangle +\langle Ah_{-\max(a)},Ah_{\max(a)}\rangle\big|.
\end{align}
First, we give out a lower bound for (\ref{e2.8}). By
\begin{align*}
\big|\langle Ah, Ah_{\max(a)}\rangle\big|
&\geq\|Ah_{\max(a)}\|_2^2-\big|\langle Ah_{-\max(a)},Ah_{\max(a)}\rangle\big|,
\end{align*}
we need to deal with $\|Ah_{\max(s)}\|_2^2$ and $\big|\langle Ah_{-\max(a)},Ah_{\max(a)}\rangle\big|$. It follows from Lemma \ref{CCLemma} that
$$
\|Ah_{\max(a)}\|_2^2\geq(1-\mu_1(a-1))\|h_{\max(a)}\|_2^2.
$$
Suppose $h=\sum_{j=1}^n c_ju^j$, where $\{c_j\}_{j=1}^n$ are nonnegative and non-increasing, and $\{u^j\}_{j=1}^n$ are indicator vectors with different supports. Then (\ref{Coneconstraintinequality}) implies that
$$
\sum_{j=s+1}^nc_j\leq\sum_{j=1}^sc_j+2\|x_{-\max(s)}\|_1.
$$
Hence,
\begin{align*}
\|h_{-\max(a)}\|_1&=\sum_{j=a+1}^sc_j+\sum_{j=s+1}^nc_j\leq\frac{s-a}{a}\sum_{j=1}^ac_j+\sum_{j=1}^sc_j+2\|x_{-\max(s)}\|_1\\
&\leq\frac{s-a}{a}\sum_{j=1}^ac_j+\frac{s}{a}\sum_{j=1}^ac_j+2\|x_{-\max(s)}\|_1=\frac{2s-a}{a}\sum_{j=1}^ac_j+2\|x_{-\max(s)}\|_1\\
&=(2s-a)\bigg(\frac{\|h_{\max(a)}\|_1}{a}+\frac{2\|x_{-\max(s)}\|_1}{2s-a}\bigg)
=:(2s-a)\alpha,
\end{align*}
and
\begin{align*}
\|h_{-\max(a)}\|_{\infty}=c_{a+1}\leq\frac{\|h_{\max(a)}\|_1}{a}\leq\frac{\|h_{\max(a)}\|_1}{a}+\frac{2\|x_{-\max(s)}\|_1}{2s-a}=\alpha.
\end{align*}
Taking $s_1=a$ and $s_2=2s-a$, then Lemma \ref{NonsparseCC} yields
\begin{align*}
\big|\langle A&h_{-\max(a)},Ah_{\max(a)}\rangle\big|\\
&\leq\alpha\sqrt{2s-a}\mu_1(a+2s-a-1)\|h_{\max(a)}\|_2\\
&\leq\sqrt{\frac{2s-a}{a}}\mu_1(2s-1)\|h_{\max(a)}\|_2^2+\sqrt{\frac{s}{2s-a}}\mu_1(2s-1)\frac{2\|x_{-\max(s)}\|_1}{\sqrt{s}}\|h_{\max(a)}\|_2.
\end{align*}
Therefore,
\begin{align}\label{e2.9}
\big|\langle A&h, Ah_{\max(a)}\rangle\big|\nonumber\\
&\geq(1-\mu_1(a-1))\|h_{\max(a)}\|_2^2\nonumber\\
&\hspace*{12pt}-\bigg(\sqrt{\frac{2s-a}{a}}\mu_1(2s-1)\|h_{\max(a)}\|_2^2
+\sqrt{\frac{s}{2s-a}}\mu_1(2s-1)\frac{2\|x_{-\max(s)}\|_1}{\sqrt{s}}\|h_{\max(a)}\|_2\bigg)\nonumber\\
&=\bigg(1-\mu_1(a-1)-\sqrt{\frac{2s-a}{a}}\mu_1(2s-1)\bigg)\|h_{\max(a)}\|_2^2\nonumber\\
&\hspace*{12pt}-\sqrt{\frac{s}{2s-a}}\mu_1(2s-1)\frac{2\|x_{-\max(s)}\|_1}{\sqrt{s}}\|h_{\max(a)}\|_2.
\end{align}

Next, we provide an upper bound of $\big|\langle Ah, Ah_{\max(a)}\rangle\big|$. Using (\ref{Tubeconstraintinequality}), we get
\begin{align}\label{e2.10}
\big|\langle A&h, Ah_{\max(a)}\rangle\big|=\big|\langle A^*Ah, h_{\max(a)}\rangle\big|\leq\|A^*Ah\|_{\infty}\|h_{\max(a)}\|_1\leq2\eta\sqrt{a}\|\|h_{\max(a)}\|_2.
\end{align}

Combining the lower bound (\ref{e2.9}) with the upper bound (\ref{e2.10}), we get
\begin{align*}
\bigg(1&-\mu_1(a-1)-\sqrt{\frac{2s-a}{a}}\mu_1(2s-1)\bigg)\|h_{\max(a)}\|_2^2\nonumber\\
&\leq\bigg(2\sqrt{a}\eta+\sqrt{\frac{s}{2s-a}}\mu_1(2s-1)\frac{2\|x_{-\max(s)}\|_1}{\sqrt{s}}\bigg)\|h_{\max(a)}\|_2.
\end{align*}
Note that
$$
\mu_1(a-1)+\sqrt{\frac{2s-a}{a}}\mu_1(2s-1)=\mu_1(a-1)+C_{a,s}\mu_1(2s-1)<1.
$$
Therefore
\begin{align}\label{e2.11}
\|h&_{\max(a)}\|_2\nonumber\\
&\leq\frac{2\sqrt{a}}{1-\mu_1(a-1)-C_{a,s}\mu_1(2s-1)}\eta
+\frac{\sqrt{s/a}\mu_1(2s-1)}{\big(1-\mu_1(a-1)-C_{a,s}\mu_1(2s-1)\big)C_{a,s}}\frac{2\|x_{-\max(s)}\|_1}{\sqrt{s}}.
\end{align}

Now, we estimate $\|h_{-max(s)}\|_2$. Applying Lemma \ref{ApproximatingerroreEtimate} with $\rho=2\|x_{-\max(s)}\|_1$ and $\beta=2$ yields
\begin{align}\label{e2.12}
\|h_{-\max(s)}\|_2\leq\|h_{\max(s)}\|_2+\frac{2\|x_{-\max(s)}\|_1}{\sqrt{s}}.
\end{align}
It follows from (\ref{e2.12}) that
\begin{align}\label{e2.13}
\|h\|_2&=\sqrt{\|h_{\max(s)}\|_2^2+\|h_{-\max(s)}\|_2^2}
\leq\sqrt{\|h_{\max(s)}\|_2^2+\bigg(\|h_{\max(s)}\|_2+\frac{2\|x_{-\max(s)}\|_1}{\sqrt{s}}\bigg)^2}\nonumber\\
&\leq\sqrt{2}\|h_{\max(s)}\|_2+\frac{2\|x_{-\max(s)}\|_1}{\sqrt{s}}
\leq\sqrt{\frac{2s}{a}}\|h_{\max(a)}\|_2+\frac{2\|x_{-\max(s)}\|_1}{\sqrt{s}}.
\end{align}
Then substituting (\ref{e2.11}) into (\ref{e2.13}), we have
\begin{align*}
\|\hat{x}^{DS}-x\|_2&\leq\frac{2\sqrt{2}\sqrt{s}}{1-\mu_1(a-1)-C_{a,s}\mu_1(2s-1)}\eta\\
&\hspace*{12pt}+\bigg(\frac{\sqrt{2}s\mu_1(2s-1)/a}{\big(1-\mu_1(a-1)-C_{a,s}\mu_1(2s-1)\big)C_{a,s}}+1\bigg)\frac{2\|x_{-\max(s)}\|_1}{\sqrt{s}}.
\end{align*}

Next, we turn our attention to consider $\hat{x}^{\ell_2}$. Let $h=\hat{x}^{\ell_2}-x$. Note that we can replace (\ref{Tubeconstraintinequality}) by
\begin{align}\label{QCBPTubeconstraintinequality}
\|Ah\|_2\leq\|A\hat{x}^{\ell_2}-b\|_2+\|b-Ax\|_2\leq\eta+\eta=2\eta.
\end{align}
And we also can replace (\ref{e2.10}) by
\begin{align}\label{e2.14}
\big|\langle A&h, Ah_{\max(a)}\rangle\big|\leq\|Ah\|_{2}\|Ah_{\max(a)}\|_2\leq2\eta\sqrt{1+\mu_1(a-1)}\|\|h_{\max(a)}\|_2.
\end{align}
Then by a similar proof of the case $\hat{x}^{DS}$, we get
\begin{align*}
\|\hat{x}^{\ell_2}-x\|_2&\leq\frac{2\sqrt{2}\sqrt{1+\mu_1(a-1)}}{1-\mu_1(a-1)-C_{a,s}\mu_1(2s-1)}\eta\\
&\hspace*{12pt}+\bigg(\frac{\sqrt{2}s\mu_1(2s-1)/a}{\big(1-\mu_1(a-1)-C_{a,s}\mu_1(2s-1)\big)C_{a,s}}+1\bigg)\frac{2\|x_{-\max(s)}\|_1}{\sqrt{s}},
\end{align*}
which finishes our conclusion.
\end{proof}

\begin{remark}\label{QCBP-DS}
If we take $a=s$ in Theorem \ref{QCBP-DSgeneral}, then $C_{a,s}=1$. And we have a simpler condition
\begin{align}\label{SpecialQCBP-DScondition-1}
\mu_1(s-1)+\mu_1(2s-1)<1.
\end{align}
Especially, under the condition
\begin{align}\label{SpecialQCBP-DScondition-2}
\mu_1(2s-1)<\frac{1}{2}
\end{align}
or
\begin{align}\label{SpecialQCBP-DScondition-3}
\mu<\frac{1}{3s-2},
\end{align}
the signal can be stably recovered from Dantzig selector (\ref{DantzigselectorMedel}) and QCBP (\ref{QCBPmodel}).
\end{remark}

\begin{remark}\label{QCBP-DS-Compare}
Our condition (\ref{SpecialQCBP-DScondition-3}) for QCBP model (\ref{QCBPmodel}) improves the condition $\mu<1/(4s-1)$ in Donoho, Elad, and Temlyakov \cite{DET2006}. And the condition (\ref{SpecialQCBP-DScondition-1}) for QCBP model (\ref{QCBPmodel}) also improves the condition $\mu_1(2s)<1/2$ for iterative hard thresholding and the condition $\mu_1(s-1)+2\mu_1(s)<1$ for hard thresholding pursuit in \cite[Chapter 5]{FR2013}.
\end{remark}

\subsection{Stable Recovery of Lasso \label{s2.3}}
\hskip\parindent

In this subsection, we consider the stable recovery of signals through Lasso model (\ref{LassoModel}).

\begin{theorem}\label{Lassogeneral}
Assume the cumulative coherence function of the measurement matrix $A$ satisfies
\begin{align}\label{Lassocondition}
\mu_1(a-1)+\mu_1(4s-1)<\frac{1}{D_{a,s}},
\end{align}
where $D_{a,s}=\sqrt{(4s-a)/a}$ and $1\leq a\leq s$, and $\|A^*z\|_{\infty}\leq\lambda/2$.  Let $\hat{x}^{L}$ be the solution to the Lasso (\ref{LassoModel}), then
\begin{align*}
\|\hat{x}^{L}-x\|_2
&\leq\frac{9(1+\mu_1(a-1))}{4\Big(1-D_{a,s}\big(\mu_1(4s-1)+\mu_1(a-1)\big)\Big)\mu_1(a-1)}\sqrt{s}\lambda\\
&\hspace*{12pt}+\bigg(\frac{8}{3\Big(1-D_{a,s}\big(\mu_1(4s-1)+\mu_1(a-1)\big)\Big)D_{a,s}}+\frac{1}{2}\bigg)
\frac{2\|x_{-\max(s)}\|_1}{\sqrt{s}}.
\end{align*}
\end{theorem}

Before giving out the proof, we first recall an auxiliary lemma. It is a modified cone constraint inequality (see, e.g., \cite[Page 2356]{CP2011} for matrix case and \cite[Lemma 2]{SHB2015} for vector with frame).

\begin{lemma}\label{LassoConeconstriantinequalityLemma}
If the noisy measurements $b=Ax+z$ are observed with noise level $\|A^*z\|_{\infty}\leq\lambda/2$, then the minimization solution $\hat{x}^{L}$ of (\ref{LassoModel}) satisfies
\begin{align*}
\|Ah\|_2^2+\lambda\|h_{-\max(s)}\|_1\leq 3\lambda\|h_{\max(s)}\|_1+4\lambda\|x_{-\max(s)}\|_1,
\end{align*}
where $h=\hat{x}^{L}-x$. In particular,
\begin{align*}
\|h_{-\max(s)}\|_1\leq 3\|h_{\max(s)}\|_1+4\|x_{-\max(s)}\|_1
\end{align*}
and
\begin{align*}
\|Ah\|_2^2\leq3\lambda\|h_{\max(s)}\|_1+4\lambda\|x_{-\max(s)}\|_1.
\end{align*}
\end{lemma}

\begin{proof}[Proof of Theorem \ref{Lassogeneral}]
Let $h=\hat{x}^{L}-x$. By Lemma \ref{LassoConeconstriantinequalityLemma}, we get a modified cone constraint inequality as follows
\begin{align}\label{LassoConeconstriantinequality}
\|h_{-\max(s)}\|_1\leq 3\|h_{\max(s)}\|_1+4\|x_{-\max(s)}\|_1,
\end{align}
and an estimate of $\|Ah\|_2^2$ as follows
\begin{align}\label{LassoTubeconstriantinequality}
\|Ah\|_2^2\leq3\lambda\|h_{\max(s)}\|_1+4\lambda\|x_{-\max(s)}\|_1.
\end{align}

We write
\begin{align*}
\|h\|_2=\sqrt{\|h_{\max(s)}\|_2^2+\|h_{-\max(s)}\|_2^2}.
\end{align*}
Then we estimate $\|h_{\max(s)}\|_2$ and $\|h_{-\max(s)}\|_2$, respectively.

Note that $\|h_{\max(s)}\|_2\leq\sqrt{s/a}\|h_{\max(a)}\|_2$. We also only need to estimate $\|h_{\max(a)}\|_2$.  we still consider the identity (\ref{e2.8}).
It follows from Lemma \ref{CCLemma} that
\begin{align*}
|\langle Ah, Ah_{\max(a)} \rangle|&\geq\|Ah_{\max(a)}\|_2^2-|\langle Ah_{\max(a)}, Ah_{-\max(a)}\rangle|\\
&\geq\big(1-\mu_1(a-1)\big)\|h_{\max(a)}\|_2^2-|\langle Ah_{\max(a)}, Ah_{-\max(a)}\rangle|.
\end{align*}
Let $h=\sum_{j=1}^nc_ju^j$, where $\{c_j\}_{j=1}^n$ is a non-negative and non-increasing sequence and $\{u^j\}_{j=1}^n$ are indicator vectors with different supports in $\mathbb{R}^n$. Then by (\ref{LassoConeconstriantinequality}), we have
$$
\sum_{j=s+1}^nc_j\leq 3\sum_{j=1}^sc_j+4\|x_{-\max(s)}\|_1.
$$
Therefore,
\begin{align*}
\|h_{-\max(a)}\|_1&=\sum_{j=a+1}^sc_j+\sum_{j=s+1}^nc_j\leq\frac{s-a}{a}\sum_{j=1}^ac_j+3\sum_{j=1}^sc_j+4\|x_{-\max(s)}\|_1\\
&\leq\frac{s-a}{a}\sum_{j=1}^ac_j+\frac{3s}{a}\sum_{j=1}^ac_j+4\|x_{-\max(s)}\|_1=\frac{4s-a}{a}\sum_{j=1}^ac_j+4\|x_{-\max(s)}\|_1\\
&=(4s-a)\bigg(\frac{\|h_{\max(a)}\|_1}{a}+\frac{4\|x_{-\max(s)}\|_1}{4s-a}\bigg)=:(4s-a)\alpha.
\end{align*}
and
\begin{align*}
\|h_{-\max(a)}\|_{\infty}=c_{a+1}\leq\frac{\|h_{\max(a)}\|_1}{a}\leq\frac{\sum_{j=1}^ac_j}{a}+\frac{4\|x_{-\max(s)}\|_1}{4s-a}=\alpha.
\end{align*}
Taking $s_1=a$ and $s_2=4s-a$ in Lemma \ref{NonsparseCC}, we get
\begin{align*}
|\langle A&h_{\max(a)}, Ah_{-\max(a)}\rangle|\\
&\leq\alpha\sqrt{4s-a}\mu_1(4s-a+a-1)\|h_{\max(a)}\|_2\\
&\leq\sqrt{\frac{4s-a}{a}}\mu_1(4s-1)\|h_{\max(a)}\|_2^2+\sqrt{\frac{s}{4s-a}}\mu_1(4s-1)\frac{4\|x_{-\max(s)}\|_1}{\sqrt{s}}\|h_{\max(a)}\|_2.
\end{align*}
Therefore, we can get a lower bound estimate of $|\langle Ah, Ah_{\max(a)} \rangle|$ as follows
\begin{align}\label{e2.15}
|\langle A&h, Ah_{\max(a)} \rangle|\nonumber\\
&\geq\big(1-\mu_1(a-1)\big)\|h_{\max(a)}\|_2^2\nonumber\\
&\hspace*{12pt}-\bigg(\sqrt{\frac{4s-a}{a}}\mu_1(4s-1)\|h_{\max(a)}\|_2^2
+\sqrt{\frac{s}{4s-a}}\mu_1(4s-1)\frac{4\|x_{-\max(s)}\|_1}{\sqrt{s}}\|h_{\max(a)}\|_2\bigg)\nonumber\\
&=\bigg(1-\mu_1(a-1)-\sqrt{\frac{4s-a}{a}}\mu_1(4s-1)\bigg)\|h_{\max(a)}\|_2^2\nonumber\\
&\hspace*{12pt}-\sqrt{\frac{s}{4s-a}}\mu_1(4s-1)\frac{4\|x_{-\max(s)}\|_1}{\sqrt{s}}\|h_{\max(a)}\|_2.
\end{align}
Next, we establish  an upper bound of $|\langle Ah, Ah_{\max(a)} \rangle|$. Using (\ref{LassoConeconstriantinequality}) and Lemma \ref{CCLemma}, we obtain
\begin{align*}
\big|\langle Ah&, Ah_{\max(a)}\rangle\big|\leq\|Ah\|_{2}\|Ah_{\max(a)}\|_{2}\nonumber\\
&\leq\sqrt{3\lambda\|h_{\max(s)}\|_1+4\lambda\|x_{-\max(s)}\|_1}\sqrt{(1+\mu_1(a-1))\|h_{\max(a)}\|_2^2}\nonumber\\
&\leq\sqrt{\frac{3\lambda s}{a}\|h_{\max(a)}\|_1+4\lambda\|x_{-\max(s)}\|_1}\sqrt{(1+\mu_1(a-1))\|h_{\max(a)}\|_2^2}\nonumber\\
&\leq\sqrt{\frac{1}{\varepsilon}\lambda\sqrt{s}\bigg(3\varepsilon\sqrt{\frac{s}{a}}\|h_{\max(a)}\|_2
+4\varepsilon\frac{\|x_{-\max(s)}\|_1}{\sqrt{s}}\bigg)}
\sqrt{1+\mu_1(a-1)}\|h_{\max(s)}\|_2,
\end{align*}
where $\varepsilon>0$ is to be determined. Then by the elementary inequality $\sqrt{|a||b|}\leq(|a|+|b|)/2$, we have that
\begin{align*}
\big|\langle Ah, Ah_{\max(a)}\rangle\big|
&\leq\frac{3\varepsilon\sqrt{\frac{s}{a}}\|h_{\max(a)}\|_2+4\varepsilon\frac{\|x_{-\max(s)}\|_1}{\sqrt{s}}+\frac{1}{\varepsilon}\sqrt{s}\lambda}{2}
\sqrt{1+\mu_1(a-1)}\|h_{\max(s)}\|_2\\
&=\frac{3\varepsilon\sqrt{\big(1+\mu_1(a-1)\big)s}}{2\sqrt{a}}\|h_{\max(a)}\|_2^2\\
&\hspace*{12pt}+\bigg(\varepsilon\sqrt{1+\mu_1(a-1)}\frac{2\|x_{-\max(s)}\|_1}{\sqrt{s}}
+\frac{\sqrt{1+\mu_1(a-1)}}{2\varepsilon}\sqrt{s}\lambda\bigg)
\|h_{\max(a)}\|_2.
\end{align*}
Taking
$$
\varepsilon=\frac{2\sqrt{a}\mu_1(a-1)}{3\sqrt{\big(1+\mu_1(a-1)\big)s}}\bigg(\sqrt{\frac{4s-a}{a}}-1\bigg),
$$
then we have
\begin{align}\label{e2.16}
\big|\langle A&h, Ah_{\max(a)}\rangle\big|\nonumber\\
&\leq\bigg(\sqrt{\frac{4s-a}{a}}-1\bigg)\mu_1(a-1)\|h_{\max(a)}\|_2^2\nonumber\\
&\hspace*{12pt}+\bigg(\frac{2\big(\sqrt{4s-a}-\sqrt{a}\big)\mu_1(a-1)}{3\sqrt{s}}\frac{2\|x_{-\max(s)}\|_1}{\sqrt{s}}
+\frac{3\sqrt{s}(1+\mu_1(a-1))}{4\big(\sqrt{4s-a}-\sqrt{a}\big)\mu_1(a-1)}\sqrt{s}\lambda\bigg)\|h_{\max(s)}\|_2.
\end{align}
Combining the lower bound (\ref{e2.15}) with the upper bound (\ref{e2.16}), we get
\begin{align*}
\bigg(1&-\sqrt{\frac{4s-a}{a}}\big(\mu_1(4s-1)+\mu_1(a-1)\big)\bigg)\|h_{\max(a)}\|_2^2\\
&\leq\bigg(\frac{2}{3}\frac{\sqrt{4s-a}-\sqrt{a}}{\sqrt{s}}\mu_1(a-1)
+2\sqrt{\frac{s}{4s-a}}\mu_1(4s-1)\bigg)\frac{2\|x_{-\max(s)}\|_1}{\sqrt{s}}\|h_{\max(s)}\|_2\\
&\hspace*{12pt}+\frac{3\sqrt{s}(1+\mu_1(a-1))}{4\big(\sqrt{4s-a}-\sqrt{a}\big)\mu_1(a-1)}\sqrt{s}\lambda\|h_{\max(s)}\|_2\\
&\leq\bigg(\frac{4}{3}\mu_1(a-1)
+\frac{2}{\sqrt{3}}\mu_1(4s-1)\bigg)\frac{2\|x_{-\max(s)}\|_1}{\sqrt{s}}\|h_{\max(s)}\|_2\\
&\hspace*{12pt}+\frac{9(1+\mu_1(a-1))}{8\mu_1(a-1)}\sqrt{s}\lambda\|h_{\max(s)}\|_2\\
&\leq\frac{4}{3}\big(\mu_1(a-1)+\mu_1(4s-1)\big)\frac{2\|x_{-\max(s)}\|_1}{\sqrt{s}}\|h_{\max(s)}\|_2
+\frac{9(1+\mu_1(a-1))}{8\mu_1(a-1)}\sqrt{s}\lambda\|h_{\max(s)}\|_2\\
\end{align*}
Note that
$$
\mu_1(a-1)+\mu_1(4s-1)<\sqrt{\frac{a}{4s-a}}=\frac{1}{D_{a,s}}.
$$
Therefore
\begin{align}\label{e2.17}
\|h_{\max(a)}\|_2
&\leq\frac{4}{3\Big(1-D_{a,s}\big(\mu_1(4s-1)+\mu_1(a-1)\big)\Big)D_{a,s}}
\frac{2\|x_{-\max(s)}\|_1}{\sqrt{s}}\nonumber\\
&\hspace*{12pt}+\frac{9(1+\mu_1(a-1))}{8\Big(1-D_{a,s}\big(\mu_1(4s-1)+\mu_1(a-1)\big)\Big)\mu_1(a-1)}
\sqrt{s}\lambda.
\end{align}

Now, we estimate $\|h_{-max(s)}\|_2$. Using (\ref{LassoConeconstriantinequality}), we have
\begin{align}\label{e2.18}
&\|h_{-\max(s)}\|_2\leq\sqrt{\|h_{-max(s)}\|_1\|h_{-\max(s)}\|_{\infty}}\nonumber\\
&\leq\sqrt{\big(3\|h_{\max(s)}\|_1+4\|x_{-\max(s)}\|_1\big)\frac{\|h_{\max(s)}\|_1}{s}}\nonumber\\
&\leq\sqrt{3\|h_{\max(s)}\|_2^2+\frac{4\|x_{-\max(s)}\|_1}{\sqrt{s}}\|h_{\max(s)}\|_2}.
\end{align}

It follows from (\ref{e2.18}) that
\begin{align}\label{e2.19}
\|h\|_2&\leq\sqrt{\|h_{\max(s)}\|_2^2+3\|h_{\max(s)}\|_2^2+\frac{4\|x_{-\max(s)}\|_1}{\sqrt{s}}\|h_{\max(s)}\|_2}\nonumber\\
&\leq2\|h_{\max(s)}\|_2+\frac{1}{2}\frac{2\|x_{-\max(s)}\|_1}{\sqrt{s}}.
\end{align}
Then substituting (\ref{e2.17}) into (\ref{e2.19}), we have
\begin{align*}
\|\hat{x}^{L}-x\|_2
&\leq\frac{9(1+\mu_1(a-1))}{4\Big(1-D_{a,s}\big(\mu_1(4s-1)+\mu_1(a-1)\big)\Big)\mu_1(a-1)}\sqrt{s}\lambda\\
&\hspace*{12pt}+\bigg(\frac{8}{3\Big(1-D_{a,s}\big(\mu_1(4s-1)+\mu_1(a-1)\big)\Big)D_{a,s}}+\frac{1}{2}\bigg)
\frac{2\|x_{-\max(s)}\|_1}{\sqrt{s}}.
\end{align*}
Therefore we finish the proof.
\end{proof}

\begin{remark}\label{Lasso}
Taking $a=s$ in Theorem \ref{Lassogeneral}, we can get a simpler condition
\begin{align}\label{SpecialLassocondition-1}
\mu_1(s-1)+\mu_1(4s-1)<\frac{1}{\sqrt{3}}.
\end{align}
Especially, under the condition
\begin{align}\label{SpecialLassocondition-2}
\mu_1(4s-1)<\frac{1}{2\sqrt{3}}
\end{align}
or
\begin{align}\label{SpecialLassocondition-3}
\mu<\frac{1}{\sqrt{3}(5s-2)},
\end{align}
the signal can be stable recovered from Lasso (\ref{LassoModel}).
\end{remark}

\begin{remark}\label{FirstCoherenceLasso}
Which should be pointed out is that it is the first time to give the coherence and cumulative coherence analysis for Lasso model (\ref{LassoModel}). And the condition (\ref{SpecialLassocondition-1}) may be improved further.
\end{remark}

\section{Prediction Loss $\|A\hat{x}^{DS}-Ax\|_2^2$ and $\|A\hat{x}^{L}-Ax\|_2^2$ \label{s3}}
\hskip\parindent

Because Lasso estimator and Dantzig selector exhibit similar behavior, we expect that the prediction loss $\|A\hat{x}^{DS}-Ax\|_2^2$ and $\|A\hat{x}^{L}-Ax\|_2^2$ is close when the number of nonzero components of the Lasso or Dantzig selector is small as compared to the same size. This question was first researched by Bickel, Ritov and Tsybakov \cite{BRT2009} under the RE-condition (see Section \ref{s5}).
And Xia and Li \cite{XL2016} also estimated the closeness of the prediction loss in the framework of robust null space property.
In this subsection, we estimate the closeness of this prediction loss in the framework of cumulative coherence.

We consider the Gaussian noise model
\begin{align}\label{Gaussiannoisemodel}
b=Ax+z,
\end{align}
where the components $z_i$ of $z$ are i.i.d. random variables with $z_i\sim \mathcal{N}(0,\sigma^2)$.
We shall assume that the noise level $\sigma$ is known.

\begin{theorem}\label{Closenessofpredictionloss}
Suppose Gaussian noise model (\ref{Gaussiannoisemodel}).  Assume the cumulative coherence function of the measurement matrix $A$ satisfies
$$
\mu_1(s-1)<1.
$$
Let $\hat{x}^{L}$ be the solution of Lasso model (\ref{LassoModel}) with $\lambda=2\sigma\sqrt{2\log n}$, and $\hat{x}^{DS}$ be the solution of Dantzig selector model (\ref{DantzigselectorMedel}) with $\eta=\lambda$. Then with probability at least
$$
1-\frac{1}{2\sqrt{\pi\log n}},
$$
we have
\begin{align*}
\big|\|A&\hat{x}^{DS}-x\|_2^2-\|A\hat{x}^{L}-Ax\|_2^2\big|
\leq8\bigg(\frac{9}{1-\mu_1(s-1)}+\frac{9}{4}\bigg)\big(\log n\big)s\sigma^2+\Bigg(\frac{2\|\hat{x}^{L}_{-\max(s)}\|_1}{\sqrt{s}}\Bigg)^2.
\end{align*}
\end{theorem}

\begin{proof}[Proof of Theorem \ref{Closenessofpredictionloss}]
By the proof of \cite[Theorem 3]{LL2014}, we know that $\|A^*(b-A\hat{x}^L)\|_{\infty}\leq\lambda=\eta$, i.e., $\hat{x}^{L}$ is also a feasible point of Dantzig selector model (\ref{DantzigselectorMedel}). Let $h=\hat{x}^{DS}-\hat{x}^{L}$. It follows from Lemma \ref{ConeconstriantinequalityLemma} that
\begin{align}\label{e3.1}
\|h_{-\max(s)}\|_1\leq\|h_{\max(s)}\|_1+2\|\hat{x}^{L}_{-\max(s)}\|_1.
\end{align}

First, we estimate $\|A\hat{x}^{L}-Ax\|_2^2$. We can rewrite $\|A\hat{x}^{L}-Ax\|_2^2$ as
\begin{align*}
\|A\hat{x}^{L}-Ax\|_2^2&=\|A\hat{x}^{DS}-Ax\|_2^2+\|Ah\|_2^2+2\langle Ah, A\hat{x}^{DS}-Ax\rangle\\
&=\|A\hat{x}^{DS}-Ax\|_2^2-\|Ah\|_2^2+2\langle Ah, A\hat{x}^{L}-Ax\rangle\\
&\leq\|A\hat{x}^{DS}-Ax\|_2^2-\|Ah\|_2^2+2\|h\|_1\|A^*(A\hat{x}^{L}-Ax)\|_{\infty}.
\end{align*}
By (\ref{e3.1}) and Lemma \ref{CCLemma}, we get
\begin{align}\label{e3.2}
\|h\|_1&=\|h_{\max(s)}\|_1+\|h_{-\max(s)}\|_1\leq2\|h_{\max(s)}\|_1+2\|\hat{x}^{L}_{-\max(s)}\|_1\nonumber\\
&\leq2\sqrt{s}\|h_{\max(s)}\|_2+2\|\hat{x}^{L}_{-\max(s)}\|_1\leq\frac{2\sqrt{s}}{\sqrt{1-\mu_1(s-1)}}\|Ah_{\max(s)}\|_2+2\|\hat{x}^{L}_{-\max(s)}\|_1.
\end{align}
And by \cite[Lemma 5.1]{CXZ2009}, we have
\begin{align}\label{e3.3}
\|A^*(A\hat{x}^{L}-Ax)\|_{\infty}&\leq\|A^*(A\hat{x}^{L}-b)\|_{\infty}+\|A^*(b-Ax)\|_{\infty}
\leq\lambda+\frac{\lambda}{2}=\frac{3\lambda}{2}
\end{align}
with probability at least
$$
1-\frac{1}{2\sqrt{\pi\log n}}.
$$
Combination of (\ref{e3.2}) and (\ref{e3.3}) yields
\begin{align}\label{e3.4}
\|A\hat{x}^{L}-Ax\|_2^2
&\leq\|A\hat{x}^{DS}-Ax\|_2^2-\|Ah\|_2^2+
3\lambda\bigg(\frac{2\sqrt{s}}{\sqrt{1-\mu_1(s-1)}}\|Ah_{\max(s)}\|_2+2\|\hat{x}^{L}_{-\max(s)}\|_1\bigg)\nonumber\\
&=\|A\hat{x}^{DS}-Ax\|_2^2+3\sqrt{s}\lambda\frac{2\|\hat{x}^{L}_{-\max(s)}\|_1}{\sqrt{s}}+
\bigg(\frac{6\sqrt{s}\lambda}{\sqrt{1-\mu_1(s-1)}}\|Ah\|_2-\|Ah\|_2^2\bigg)\nonumber\\
&\leq\|A\hat{x}^{DS}-Ax\|_2^2+3\sqrt{s}\lambda\frac{2\|\hat{x}^{L}_{-\max(s)}\|_1}{\sqrt{s}}+\frac{9s\lambda^2}{1-\mu_1(s-1)},
\end{align}
where the last inequality depends on the elementary inequality $ab-b^2\leq a^2/4$ for all $a, b\in\mathbb{R}$.

Next, we estimate $\|A\hat{x}^{DS}-Ax\|_2^2$. We can rewrite $\|A\hat{x}^{DS}-Ax\|_2^2$ as
\begin{align*}
\|A\hat{x}^{DS}-Ax\|_2^2&=\|A\hat{x}^{L}-Ax\|_2^2+\|Ah\|_2^2+2\langle Ah, A\hat{x}^{L}-Ax\rangle\\
&=\|A\hat{x}^{L}-Ax\|_2^2-\|Ah\|_2^2+2\langle Ah, A\hat{x}^{DS}-Ax\rangle\\
&\leq\|A\hat{x}^{L}-Ax\|_2^2-\|Ah\|_2^2+2\|h\|_1\|A^*(A\hat{x}^{DS}-Ax)\|_{\infty}.
\end{align*}
By \cite[Lemma 5.1]{CXZ2009}, we have
\begin{align}\label{e3.5}
\|A^*(A\hat{x}^{DS}-Ax)\|_{\infty}&\leq\|A^*(A\hat{x}^{DS}-b)\|_{\infty}+\|A^*(b-Ax)\|_{\infty}
\leq\eta+\frac{\lambda}{2}=\frac{3\eta}{2}
\end{align}
with probability at least
$$
1-\frac{1}{2\sqrt{\pi\log n}}.
$$
Combination of (\ref{e3.2}) and (\ref{e3.5}) yields
\begin{align}\label{e3.6}
\|A\hat{x}^{DS}-Ax\|_2^2
&\leq\|A\hat{x}^{L}-Ax\|_2^2-\|Ah\|_2^2+
3\eta\bigg(\frac{2\sqrt{s}}{\sqrt{1-\mu_1(s-1)}}\|Ah_{\max(s)}\|_2+2\|\hat{x}^{L}_{-\max(s)}\|_1\bigg)\nonumber\\
&=\|A\hat{x}^{L}-Ax\|_2^2+3\sqrt{s}\eta\frac{2\|\hat{x}^{L}_{-\max(s)}\|_1}{\sqrt{s}}+
\bigg(\frac{6\sqrt{s}\eta}{\sqrt{1-\mu_1(s-1)}}\|Ah\|_2-\|Ah\|_2^2\bigg)\nonumber\\
&\leq\|A\hat{x}^{L}-Ax\|_2^2+3\sqrt{s}\eta\frac{2\|\hat{x}^{L}_{-\max(s)}\|_1}{\sqrt{s}}+\frac{9s\eta^2}{1-\mu_1(s-1)}.
\end{align}

Finally, by observations (\ref{e3.4}) and (\ref{e3.6}), we get
\begin{align*}
\big|\|A\hat{x}^{DS}-Ax\|_2^2-\|A\hat{x}^{L}-Ax\|_2^2\big|
&\leq3\sqrt{s}\eta\frac{2\|\hat{x}^{L}_{-\max(s)}\|_1}{\sqrt{s}}+\frac{9s\eta^2}{1-\mu_1(s-1)}\\
&\leq\bigg(\frac{9}{1-\mu_1(s-1)}+\frac{9}{4}\bigg)s\eta^2+\Bigg(\frac{2\|\hat{x}^{L}_{-\max(s)}\|_1}{\sqrt{s}}\Bigg)^2,
\end{align*}
where the last inequality follows from the elementary inequality $2|ab|\leq a^2+b^2$ for all $a, b\in\mathbb{R}$.
\end{proof}

\section{Oracle Inequality \label{s4}}
\hskip\parindent

The oracle inequality approach was introduced by Donoho and Johnstone \cite{DJ1994} in the context of wavelet thresholding for signal denoising. It provides an effective tool for studying the performance of an estimation procedure by comparing it to that of an ideal estimator. This approach has been extended to study compressed sensing by Cand\`{e}s and Tao's groundbreaking work \cite{CT2007}.  In \cite{CT2007}, they developed an oracle inequality for Dantzig selector $\hat{x}^{DS}$ in the Gaussian noise setting in the framework of restricted isometry property. Later, Cand\`{e}s and Plan \cite{CP2011} extended it to matrix Lasso and matrix Dantzig selector under the condition of restricted isometry property. And almost at the same time, Cai, Wang and Xu \cite{CWX2010} extended it to Dantzig selector $\hat{x}^{DS}$ for sparse signals in the framework of mutual incoherence property. Moreover works about oracle inequality, readers can see \cite{BRT2009,C2013,CZ2014}. Motivated by \cite{CP2011} and \cite{CWX2010}, we consider oracle inequality  under the framework of cumulative coherence.

Before stating our main results,  we first give two notations. Let
$$
S_0=\{j\in[[1,n]]: |x(j)|\geq\sigma\},
$$
\begin{align}\label{e4.1}
G(\xi,x)=\sigma^2\|\xi\|_0+\|x-\xi\|_2^2.
\end{align}
and
\begin{align}\label{e4.2}
H(\xi,x)=\iota\|\xi\|_0+\|Ax-A\xi\|_2^2,
\end{align}
where $\iota=\lambda^2/8$.

\begin{theorem}\label{OracleinequalityDS}
Consider the Gaussian noise model (\ref{Gaussiannoisemodel}). Let $x$ be $s$-sparse.
 Suppose that the cumulative coherence function of the measurement matrix $A\in\mathbb{R}^{m\times n}$ satisfies
 $$
 \mu_1(s-1)+\mu_1(2s-1)<1.
 $$
 Set $\eta^*=3\sigma \sqrt{2\log n}/2$. Let $\hat{x}^{DS}$ be the minimizer of the problem
 \begin{align}\label{GaussiannoiseDSModel}
 \min_{y\in\mathbb{R}^n}~\|y\|_1 ~~\text{subject~ to}~~\|A^*(b-Ay)\|_{\infty}\leq\eta^*.
 \end{align}
 Then with probability at least
 $$
 1-\frac{1}{2\sqrt{\pi\log n}},
 $$
 $\hat{x}^{DS}$ satisfies
 \begin{align*}
\|\hat{x}^{DS}-x\|_2^2
&\leq\frac{72(5+\log n)}{\big(1-\mu_1(s-1)-\mu_1(2s-1)\big)^2}\sum_{j}\max\{\sigma^2,|x(j)|^2\}.
\end{align*}
\end{theorem}

In order to prove Theorem \ref{OracleinequalityDS}, we give one useful lemma. Its proof follows the same lines of the proof of \cite[Lemma 3.5]{CP2011} and we omit it.
\begin{lemma}\label{Feasiblepoint}
Let $\bar{x}=\arg\min_{\xi\in\mathbb{R}^n}H(\xi,x)$, then $\|A^*(A\bar{x}-Ax)\|_{\infty}\leq\lambda/2$.
\end{lemma}

Now we begin to prove Theorem \ref{OracleinequalityDS}.

\begin{proof}[Proof of Theorem \ref{OracleinequalityDS}]
Set $\lambda=\sigma\sqrt{2\log n}$. By \cite[Lemma 5.1]{CXZ2009}, event $E=\{z\in\mathbb{R}^m: \|A^*z\|_{\infty}\leq\lambda\}$ occurs with probability at least
$$
1-\frac{1}{2\sqrt{\pi\log n}}.
$$
In the following, we shall assume that event $E$ occurs.

We can rewrite $\|\hat{x}^{DS}-x\|_2^2$ as
\begin{align}\label{e4.5}
\|\hat{x}^{DS}-x\|_2^2\leq2\|\hat{x}^{DS}-\bar{x}\|_2^2+2\|x-\bar{x}\|_2^2.
\end{align}
Next we estimate $\|\hat{x}^{DS}-\bar{x}\|_2^2$ and $\|x-\bar{x}\|_2^2$, respectively.

First, we estimate $\|\hat{x}^{DS}-\bar{x}\|_2^2$. It follows from \cite[Lemma 5.1]{CXZ2009} and Lemma \ref{Feasiblepoint} that
\begin{align}\label{e4.6}
\|A^*(b-A\bar{x})\|_{\infty}\leq\|A^*(b-Ax)\|_{\infty}+\|A^*(Ax-A\bar{x})\|_{\infty}\leq\lambda+\frac{\lambda}{2}=\frac{3\lambda}{2}=\eta^*
\end{align}
with probability at least
$$
1-\frac{1}{2\sqrt{\pi\log n}}.
$$
Therefore, $\bar{x}$ is a feasible point of Dantzig selector model (\ref{GaussiannoiseDSModel}). And by the definition of $H(\xi,x)$, we have $H(\bar{x},x)\leq H(x,x)$, which implies that $\|\bar{x}\|_0\leq\|x\|_0\leq s$. Thus plugging $\bar{x}$ into Theorem \ref{QCBP-DS} gives
\begin{align}\label{e4.7}
\|\hat{x}^{DS}-\bar{x}\|_2^2&\leq\bigg(\frac{2\sqrt{2}}{1-\mu_1(\|\bar{x}\|_0-1)-\mu_1(2\|\bar{x}\|_0-1)}\sqrt{\|\bar{x}\|_0}\eta^*\bigg)^2\nonumber\\
&\leq\frac{8(\eta^*)^2}{\big(1-\mu_1(s-1)-\mu_1(2s-1)\big)^2}\|\bar{x}\|_0.
\end{align}

Now, we turn our attention to estimate $\|x-\bar{x}\|_2^2$. Plugging $\bar{x}-x$ into Lemma \ref{CCLemma} yields
\begin{align}\label{e4.8}
\|\bar{x}-x\|_2^2\leq\frac{1}{1-\mu_1(2s-1)}\|A\bar{x}-Ax\|_2^2.
\end{align}

Then substituting (\ref{e4.7}) and (\ref{e4.8}) into (\ref{e4.5}), we get
\begin{align*}
\|\hat{x}^{DS}-x\|_2^2&\leq\frac{16(\eta^*)^2}{\big(1-\mu_1(s-1)-\mu_1(2s-1)\big)^2}\|\bar{x}\|_0+\frac{2}{1-\mu_1(2s-1)}\|A\bar{x}-Ax\|_2^2\\
&=\frac{16(\eta^*)^2}{\big(1-\mu_1(s-1)-\mu_1(2s-1)\big)^2\iota}\iota\|\bar{x}\|_0+\frac{2}{1-\mu_1(2s-1)}\|A\bar{x}-Ax\|_2^2\\
&\leq\frac{16(\eta^*)^2}{\big(1-\mu_1(s-1)-\mu_1(2s-1)\big)^2\iota}H(\bar{x},x)\\
&=\frac{144\big(1+\mu_1(1)\big)}{\big(1-\mu_1(s-1)-\mu_1(2s-1)\big)^2}H(\bar{x},x).
\end{align*}

Note that $H(\bar{x},x)\leq H(x_{S_0},x)$. It suffices to deal with $H(x_{S_0},x)$.
Notice that
\begin{align}\label{e4.3}
G(x_{S_0},x)&=\sigma^2\|x_{S_0}\|_0+\|x_{S_0^c}\|_2^2=\sigma^2\sum_{j\in S_0}\chi_{S_0}(j)+\sum_{j\in S_0^c}|x(j)|^2\nonumber\\
&=\sum_{j}\min\{\sigma^2,|x(j)|^2\}.
\end{align}
Therefore, an application of (\ref{e4.3}) yields
\begin{align}\label{e4.9}
H(x_{S_0},x)&=\frac{\lambda^2}{4\big(1+\mu_1(1)\big)}\|x_{S_0}\|_0+\|Ax_{S_0^c}\|_2^2\nonumber\\
&\leq\frac{\log n}{2\big(1+\mu_1(1)\big)}\sigma^2\|x_{S_0}\|_0+\big(1+\mu_1(s-1)\big)\|x_{S_0^c}\|_2^2\nonumber\\
&\leq\bigg(\frac{\log n}{2\big(1+\mu_1(1)\big)}+\big(1+\mu_1(s-1)\big)\bigg)G(x_{S_0},x)\nonumber\\
&\leq\frac{5+\log n}{2\big(1+\mu_1(1)\big)}\sum_{j}\max\{\sigma^2,|x(j)|^2\},
\end{align}
where the last inequality follows from $\mu_1(1)\leq\mu_1(2s-1)$ and $\mu_1(s-1)+\mu_1(2s-1)<1$, and the inequality $ab\leq(a+b)^2/4$.

Therefore, by $H(\bar{x},x)\leq H(x_{S_0},x)$ and (\ref{e4.9}), we have
\begin{align*}
\|\hat{x}&^{DS}-x\|_2^2\\
&\leq\frac{144\big(1+\mu_1(1)\big)}{\big(1-\mu_1(s-1)-\mu_1(2s-1)\big)^2}\frac{5+\log n}{2\big(1+\mu_1(1)\big)}\sum_{j}\max\{\sigma^2,|x(j)|^2\}\\
&=\frac{72(5+\log n)}{\big(1-\mu_1(s-1)-\mu_1(2s-1)\big)^2}\sum_{j}\max\{\sigma^2,|x(j)|^2\},
\end{align*}
which finishes the proof.
\end{proof}

\begin{remark}
It is a pity that our method of proving the oracle inequality for $\hat{x}^{DS}$ may not hold for $\hat{x}^{L}$.
\end{remark}

%%%%%%%%%%%%%%%%%%%%%%%%%%%%%%%%%%%%%%%%%%%%%%%%%%%%%%%%%%%%%%%%
%%%%%%%%%%%%%%%%%%%%%%%  Section 5 %%%%%%%%%%%%%%%%%%%%%%%%%%%%%%
%%%%%%%%%%%%%%%%%%%%%%%%%%%%%%%%%%%%%%%%%%%%%%%%%%%%%%%%%%%%%%%%
\section{Relationship to Restricted Eigenvalue Condition \label{s5}}
\hskip\parindent

In \cite{BRT2009}, Bickel, Ritov and Tsybakov introduced a key assumption-restricted eigenvalue condition (RE-condition), which is needed to guarantee nice statistical properties of the Lasso and Dantzig selectors. Many researchers have investigated the RE-condition, especial the relationship between some other recovery condition and the RE-condition. For example, \cite{C2013} showed that RE-condition implies the robust null space property.  Xia and Li \cite{XL2016} illustrated the relationship between the frame restricted isometry property (RIP) and the frame RE-condition, and showed that the frame RE-condition is a relaxation of the frame RIP.

In this section, we will investigate the relationship between cumulative coherence and the RE-condition.
First, we recall the restricted eigenvalue condition.

\begin{definition}\label{REcondition}
Let $1\leq s\leq n$ and $\tau>0$. A measurement matrix $A\in\mathbb{R}^{m\times n}$ satisfies the RE-condition of order $s$ and $\tau$ with constant $K(s,\tau, A)=K(s,\tau)$, if for all $0\neq x\in\mathbb{R}^n$,
\begin{align*}
K(x,\tau):=\min_{S\subset [n]:|S|\leq s}\min_{\|x_{S^c}\|_1\leq\tau\|x_{S}\|_1}\frac{\|Ax\|_2}{\|x_{S}\|_2}>0.
\end{align*}
\end{definition}

\begin{theorem}\label{CC-REcondition}
Suppose that the cumulative coherence function of the measurement matrix $A\in\mathbb{R}^{m\times n}$ satisfies
$$
\mu_1(s-1)+\tau\sqrt{s}\mu_1(s)<1,
$$
then matrix $A$ satisfies the RE-condition with
$$
K(s,\tau)\geq\frac{\bigg(1-\mu_1(s-1)-\tau\sqrt{s}\mu_1(s)\bigg)}{\sqrt{(1+\mu_1(s-1)}}.
$$
And if the cumulative coherence function of the measurement matrix $A\in\mathbb{R}^{m\times n}$ satisfies
$$
\mu_1(s+a-1)+\tau\sqrt{\frac{s}{b}}\mu_1(s+a+b-1)<1,
$$
where $1\leq b\leq 4a$, then matrix $A$ satisfies the RE-condition with
$$
K(s,\tau)\geq\frac{\bigg(1-\mu_1(s+a-1)-\tau\sqrt{\frac{s}{b}}\mu_1(s+a+b-1)\bigg)}{\sqrt{(1+\mu_1(s+a-1)}}.
$$
\end{theorem}

The following lemma comes from \cite{CWX2010-1}.

\begin{lemma}\label{2normestimate}
For any $x\in\mathbb{R}^n$
$$
\|x\|_2-\frac{\|x\|_1}{\sqrt{n}}\leq\frac{\sqrt{n}}{4}\big(\max_{1\leq j\leq n}|x_j|-\min_{1\leq j\leq n}|x_j|\big).
$$
\end{lemma}

\begin{proof}[Proof of Theorem \ref{CC-REcondition}]
Let $M(s,\tau):=\{x\in\mathbb{R}^n: \exists T\subset [n]~\text{and}~|T|\leq s~\text{s.t.}~\|x_{T^c}\|_1\leq\tau\|x_{T}\|_1\}$. Let $x\in M(s,\tau)$ and $\|x\|_2=1$. Then there exists $T\subset[n]$ with $|T|\leq s$ such that $\|x_{T^c}\|_1\leq\tau\|x_{T}\|_1$. We take $S$ as the locations of the $|T|$ largest coefficients of $x$ in magnitude. And denote $S_0$ as the locations of the $a$ largest coefficients of $x_{S^c}$ in magnitude. Denote $S_{*}=S\cup S_0$. We decompose $S_{*}^c$ as
$$
S_{*}^c=\cup_{j\geq} S_j,
$$
where $S_1$ is the index set of the $b$ largest entries of $x_{S_{*}^c}$, $S_2$ is the index set of the $b$ largest entries of $x_{(S_*\cup S_1)^c}$, and so on.

We consider the following identity
\begin{align}\label{e5.1}
|\langle Ax, Ax_{S_*}\rangle|=|\langle Ax_{S_*}, Ax_{S_*}\rangle-\langle Ax_{S_*^c}, Ax_{S_*}\rangle|.
\end{align}
First, we give out a lower bound for (\ref{e5.1}).  It follows from Lemma \ref{CCLemma} and Lemma \ref{OrthogonalCC} that
\begin{align}\label{e5.2}
|\langle Ax, Ax_{S_*}\rangle|&\geq\|Ax_{S_*}\|_2^2-\sum_{j\geq 1}|\langle Ax_{S_j}, Ax_{S_*}\rangle|\nonumber\\
&\geq\big(1-\mu_1(s+a-1)\big)\|x_{S_*}\|_2^2-\mu_1(s+a+b-1)\sum_{j\geq 1}\|x_{S_j}\|_2\|x_{S_*}\|_2.
\end{align}
For the case $1\leq b\leq 4a$, Lemma \ref{2normestimate}  and $\|x_{T^c}\|_1\leq\tau\|x_{T}\|_1$ yield
\begin{align}\label{e5.3}
\sum_{j\geq 1}\|x_{S_j}\|_2&\leq\frac{1}{\sqrt{b}}\sum_{j\geq 1}\|x_{S_j}\|_1+\frac{\sqrt{b}}{4}\|x_{S_1}\|_\infty
\leq\frac{1}{\sqrt{b}}\bigg(\sum_{j\geq1}\|x_{S_j}\|_1+\frac{b}{4a}\|x_{S_0}\|_1\bigg)\nonumber\\
&\leq\frac{1}{\sqrt{b}}\|x_{S^c}\|_1\leq\frac{1}{\sqrt{b}}\|x_{T^c}\|_1\leq\frac{\tau}{\sqrt{b}}\|x_{T}\|_1
\leq\frac{\tau}{\sqrt{b}}\|x_{S}\|_1\leq\tau\sqrt{\frac{s}{b}}\|x_{S}\|_2.
\end{align}
For the case $a=0$ and $b=1$, we have a more simpler estimate
\begin{align}\label{e5.4}
\sum_{j\geq 1}\|x_{S_j}\|_2=\sum_{j\geq 1}\|x_{S_j}\|_1=\|x_{S^c}\|_{1}\leq\tau\|x_{S}\|_1\leq\tau\sqrt{s}\|x_{S}\|_2.
\end{align}
Then substituting (\ref{e5.3}) or (\ref{e5.4}) into (\ref{e5.2}), we get
\begin{align}\label{e5.5}
|\langle Ax, Ax_{S_*}\rangle|
&\geq\big(1-\mu_1(s+a-1)\big)\|x_{S_*}\|_2^2-\tau\sqrt{\frac{s}{b}}\mu_1(s+a+b-1)\|x_{S}\|_2\|x_{S_*}\|_2\nonumber\\
&\geq\bigg(1-\mu_1(s+a-1)-\tau\sqrt{\frac{s}{b}}\mu_1(s+a+b-1)\bigg)\|x_{S_*}\|_2^2
\end{align}
or
\begin{align}\label{e5.6}
|\langle Ax, Ax_{S_*}\rangle|
&\geq\big(1-\mu_1(s-1)\big)\|x_{S_*}\|_2^2-\tau\sqrt{s}\mu_1(s)\|x_{S}\|_2\|x_{S_*}\|_2\nonumber\\
&\geq\bigg(1-\mu_1(s-1)-\tau\sqrt{s}\mu_1(s)\bigg)\|x_{S_*}\|_2^2.
\end{align}

Next, we give an upper bound for (\ref{e5.1}).  It follows from Lemma \ref{CCLemma} that
\begin{align}\label{e5.7}
|\langle Ax, Ax_{S_*}\rangle|\leq\|Ax_{S_*}\|_2\|Ax\|_2\leq\sqrt{(1+\mu_1(s+a-1)}\|x_{S_*}\|_2\|Ax\|_2.
\end{align}

For the case $a\geq 1$ and $b\leq 4a$, we combine (\ref{e5.5}) with (\ref{e5.7}) to get
\begin{align*}
\sqrt{(1+\mu_1(s+a-1)}\|Ax\|_2&\geq\bigg(1-\mu_1(s+a-1)-\tau\sqrt{\frac{s}{b}}\mu_1(s+a+b-1)\bigg)\|x_{S_*}\|_2\\
&\geq\bigg(1-\mu_1(s+a-1)-\tau\sqrt{\frac{s}{b}}\mu_1(s+a+b-1)\bigg)\|x_{S}\|_2,
\end{align*}
which implies that
$$
K(s,\tau)\geq\frac{\bigg(1-\mu_1(s+a-1)-\tau\sqrt{\frac{s}{b}}\mu_1(s+a+b-1)\bigg)}{\sqrt{(1+\mu_1(s+a-1)}}.
$$
And for the case $a=0$ and $b=1$, we combine (\ref{e5.6}) with (\ref{e5.7}) to get
\begin{align*}
\sqrt{(1+\mu_1(s-1)}\|Ax\|_2&\geq\bigg(1-\mu_1(s-1)-\tau\sqrt{s}\mu_1(s)\bigg)\|x_{S_*}\|_2\\
&\geq\bigg(1-\mu_1(s-1)-\tau\sqrt{s}\mu_1(s)\bigg)\|x_{S}\|_2,
\end{align*}
which implies that
$$
K(s,\tau)\geq\frac{\bigg(1-\mu_1(s-1)-\tau\sqrt{s}\mu_1(s)\bigg)}{\sqrt{(1+\mu_1(s-1)}}.
$$
\end{proof}

%%%%%%%%%%%%%%%%%%%%%%%%%%%%%%%%%%%%%%%%%%%%%%%%%%%%%%%%%%%%%%%%
%%%%%%%%%%%%%%%%%%%%%%%  Section 6 %%%%%%%%%%%%%%%%%%%%%%%%%%%%%%
%%%%%%%%%%%%%%%%%%%%%%%%%%%%%%%%%%%%%%%%%%%%%%%%%%%%%%%%%%%%%%%%
\section{Numerical Experiments \label{s6}}
\hskip\parindent

In this section, we present several numerical experiments for three different measurement matrices to support our theory. We use IRucLq-v algorithm ($q=1$) in \cite{LXY2013} to compute unconstraint problem Lasso (\ref{LassoModel}).

\subsection{Three Examples of Measurement Matrices \label{s6.1}}
\hskip\parindent

In this subsection, we will exhibit several different measurement matrices, which also are called dictionaries, and compute their coherence $\mu$ and cumulative coherence $\mu_1(s)$. And in next subsection, we will use these measurement matrices in our numerical experiments.

\begin{example}\label{DecayingAtoms}
Fix a parameter $0<\beta <1$. For each index $i\geq 0$, define an atom by
\begin{align*}
A_{i}(j)=
\begin{cases}
0, & 0\leq j<i\\
\beta^{j-i}\sqrt{1-\beta^2}, & i\leq j\leq n-1.
\end{cases}
\end{align*}
It can be shown that the atoms $\{A_i\}$ span $\ell_2(\mathbb{R}^m)$, so they form a dictionary. $A=[A_1,\dots,A_n]$ is called decaying matrix or decaying dictionary in \cite{T2004-1,T2004-2}. Each $A_i$ has $\ell_2$ unit norm. It also follows that the coherence of the measurement matrix equals $\beta$. However, the cumulative coherence function $\mu_1(s)<2\beta/(1-\beta)$
for all $s$. On the other hand, the quantity $s\mu$ grows without bound.

Set $\beta=1/(4\sqrt{3}+1)$, then the  MIP condition (\ref{SpecialLassocondition-3}) in Remark \ref{Lasso} requires that $s<2$. On the other hand, $\mu_1(4s-1)<1/(2\sqrt{3})$ for every $s$. Therefore, cumulative condition (\ref{SpecialLassocondition-2}) in Remark \ref{Lasso} shows that Lasso also can recover any (finite) linear combination of decaying atoms.
\end{example}

\begin{example}\label{DiracHadamard}
Let $A=[I_m, \bar{H}_m]$, concatenating two orthonormal bases-the standard and Hadamard bases for signals of length $m$, which is called Dirac-Hadamard matrix or Dirac-Hadamard dictionary \cite{DET2006}. Hadamard matrix $H_m$ is a square matrix whose entries are either $+1$ or $-1$ and whose rows are mutually orthogonal. And we use the normalized Hadamard matrix $\bar{H}_m=H_m/{\sqrt{m}}$. As showed in \cite{DET2006}, the coherence $\mu=1/{\sqrt{m}}$. Note that
\begin{align*}
|\langle A_i, A_k\rangle|=
\begin{cases}
|\langle I_m(i), I_m(k)\rangle|=0,&i,k\in \{1,\ldots,m\}\\
|\langle I_m(i), \bar{H}_m(k)\rangle|=\frac{1}{\sqrt{m}},&i\in \{1,\ldots,m\}, j\in\{m+1,\ldots,2m\}\\
|\langle \bar{H}_m(i), \bar{H}_m(k)\rangle|=0,&i,k\in \{m+1,\ldots,2m\}.
\end{cases}
\end{align*}
Therefore $\mu_1(s)=s/{\sqrt{m}}$.
\end{example}

\begin{example}\label{DiracFourier}(\cite{T2004-1,HGT2006})
Consider the dictionary for $\mathbb{C}^m$ that has synthesis matrix $A=[I_m, \mathcal{F}_m]$,
where $\mathcal{F}_m$ is the $m$-dimensional discrete Fourier transform matrix. For reference, the $(j,k)$ entry of $\mathcal{F}_m$ is the complex number $\exp\{-2\pi jk/m \}$, where $i$ satisfies $i^2=-1$. This dictionary is called the Dirac-Fourier dictionary or Dirac-Fourier because it consists of impulses and discrete complex exponentials.  It is very easy to check that the coherence $\mu$ of the Dirac-Fourier dictionary is $1/{\sqrt{m}}$. And by similar discussion in  Example \ref{DiracHadamard}, we get $\mu_1(s)=s/{\sqrt{m}}$.
\end{example}

\subsection{Numerical Experiments for Three Matrices \label{s6.2}}
\hskip\parindent

In this subsection we will use IRucLq-v algorithm ($q=1$) \cite{LXY2013} to solve unconstraint problem (\ref{LassoModel}) to support our theory. We consider both noiseless and noisy cases. In
this test,  the true vector $x^0$ had $s$ nonzeros with each one entry generated according to the standard Gaussian distribution and $s$ varying among $\{2, 4, 6, \ldots, 32\}$. The location of nonzeros was uniformly randomly generated. We take $A$ as three matrices-decaying matrix, Dirac-Hadamard matrix and Dirac-Fourier matrix, which has the size of $64\times 128$.
And the measurement vector $b$ was observed from $b= Ax^0+z$, where $z$ was zero-mean Gaussian noise with standard deviation $\sigma$ or zero vector. The parameter $\lambda$ was set to $10^{-6}$. We let the algorithm run to 500 iterations. The
recovery was regarded as successful if $\frac{\|x^r-x^0\|_2}{\|x^0\|_2}\leq 10^{-3}$, where $x^r$ stands for a recovered
vector.

In the noiseless case, we compare three measurement matrices in terms of success percentage. We run 50 independent realizations and record the corresponding success rates at various sparsity levels $s$. The left picture in Figure \ref{4.1}	shows these results.
From the figure, we can see that sparse signal can be exact recovered by three matrices-decaying matrix, Dirac-Hadamard matrix and Dirac-Fourier matrix, which satisfies our cumulative coherence condition in Remark \ref{Lasso}. And in three measurement matrices, Dirac-Fourier matrix gives the highest successful rate.

In the presence of noise, we take $\sigma=0.01$ and draw up the average reconstruction signal to noise ratio (SNR) over 50 experiments. The SNR is given by $\text{SNR}(x^r,x^0)=10\log_{10}\frac{\|x^r-x^0\|_2}{\|x^0\|_2}$ where the measure of the SNR is dB.  The right picture in Figure \ref{4.1} shows the SNR of stable recovery using IRucLq-v algorithm over 50 independent trials for various matrices $A$ and sparsity levels $s$. From the figure, we can see that sparse signal can be stable recovered by three matrices-decaying matrix, Dirac-Hadamard matrix and Dirac-Fourier matrix, which satisfies our cumulative coherence condition in Remark \ref{Lasso}. And in three measurement matrices, Dirac-Fourier matrix gives the smallest SNR.

\begin{figure}[H]
\begin{centering}
\subfloat[]
{\begin{centering}
\includegraphics[width=6.0cm]{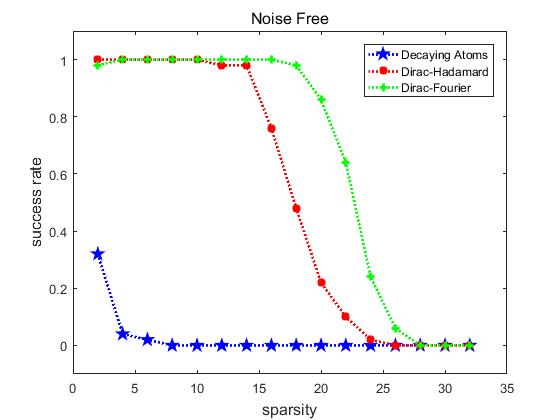}
\par\end{centering}
}
\subfloat[]
{\begin{centering}
\includegraphics[width=6.0cm]{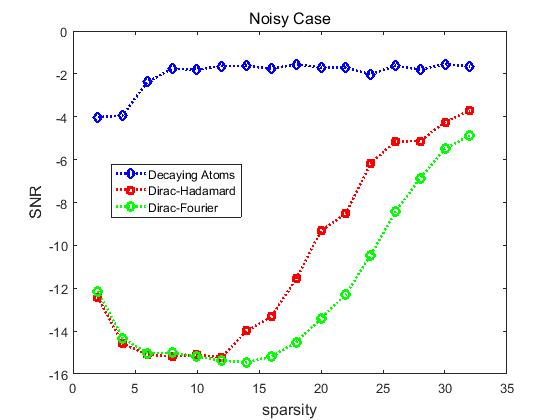}
\par\end{centering}
}
\par\end{centering}
\caption{Left: Success rates using decaying matrix, Dirac-Hadamard matrix and Dirac-Fourier matrix with $m=64, n=
128, s=2,4,6,\ldots,32$. Right: SNR using decaying matrix, Dirac-Hadamard matrix and Dirac-Fourier matrix with $m=64, n=
128, s=2,4,6,\ldots,32$.}
\label{4.1}	
\end{figure}

%%%%%%%%%%%%%%%%%%%%%%%%%%%%%%%%%%%%%%%%%%%%%%%%%%%%%%%%%%%%%%%%%%%%%%%%%%%%%%%%%%%%%%%%%%%%%%%%%%%%
%%%%%%%%%%%%%%%%%%%%%%%%%%%%%%%%%%%%%%% Section 7 %%%%%%%%%%%%%%%%%%%%%%%%%%%%%%%%%%%%%%%%%%%%%
%%%%%%%%%%%%%%%%%%%%%%%%%%%%%%%%%%%%%%%%%%%%%%%%%%%%%%%%%%%%%%%%%%%%%%%%%%%%%%%%%%%%%%%%%%%%%%%%%%%%
\section{Conclusions and Discussion \label{s7}}
\hskip\parindent

In this paper, we first show that the condition $\mu_1(s-1)+\mu_1(2s-1)<1$ can guarantee the stable recovery of the original signal $x$ via the QCBP model (\ref{QCBPmodel}) and Dantzig selector model (\ref{DantzigselectorMedel}) (Theorem \ref{QCBP-DS}). And we also show that the condition $\mu_1(s-1)+\mu_1(4s-1)<1/\sqrt{3}$ is sufficient to guarantee the stable recovery of the original signal $x$ via Lasso model (\ref{LassoModel}). Because Lasso estimator and Dantzig selector exhibit similar behavior, we also prove that the prediction loss $\|A\hat{x}^{DS}-Ax\|_2^2$ and $\|A\hat{x}^{L}-Ax\|_2^2$
is close when the number of nonzero components of the Lasso or Dantzig selector is small as compared to the same size (Theorem \ref{Closenessofpredictionloss}). For Dantzig selector model, we provide an oracle inequality for sparse signal under the condition $\mu_1(s-1)+\mu_1(2s-1)<1$ (Theorem \ref{OracleinequalityDS}).

And in the section \ref{s5}, we investigate the relationship between cumulative coherence and RE-condition, we find that the RE-condition of order $s$ and $\tau$ can be deduced from the condition $\mu_1(s-1)+\tau\sqrt{s}\mu_1(s)<1$ (Theorem \ref{CC-REcondition}). In the last section, we present several numerical experiments through IRucLq-v algorithm ($q=1$) for three measurement matrices to support our cumulative coherence theory proposed in this paper.

However, Tropp \cite{T2004-2} showed that the condition $\mu_1(s-1)+\mu_1(s)<1$ is sufficient to guarantee the exact recovery of all $s$-sparse signal. Therefore, our condition $\mu_1(s-1)+\mu_1(2s-1)<1$ for QCBP model and Dantzig selector model (Theorem \ref{QCBP-DS}) may be improved further.

%%%%%%%%%%%%%%%%%%%%%%%%%%%%%%%%%%%%%%%%%%%%%%%%%%%%%%%%%%%%%%%
%%%%%%%%%%%%%%%%%%%%%%%%%%%%%%%%%%%%%%%%%%%%%%%%%%%%%%%%%%%%%%
%%%%%%%%%%%%%%%%%%%%%%%%%%%%%%%%%%%%%%%%%%%%%%%%%%%%%%%%%%%%%%

\textbf{Acknowledgement}: Wengu Chen is supported by National Natural Science Foundation of China (No. 11371183).

%%%%%%%%%%%%%%%%%%%%%%%%%%%%%%%%%%%%%%%%%%%%%%%%%%%%%%%%%%%%%%%%%%%%%%%%%%%%%%%%%%%%%%%%%%%

%%%%%%%%%%%%%%%%%%%%%%%%%%%%%%%%%%%   Bibliography  %%%%%%%%%%%%%%%%%%%%%%%%%%%%%%%%%%%%%%%%

%%%%%%%%%%%%%%%%%%%%%%%%%%%%%%%%%%%%%%%%%%%%%%%%%%%%%%%%%%%%%%%%%%%%%%%%%%%%%%%%%%%%%%%%%%%%

\end{document}